\documentclass{elsarticle}

\usepackage{lineno,hyperref}
\modulolinenumbers[5]

\usepackage[american]{babel}	%English hyphenation etc.
\usepackage[utf8]{inputenc}		%Input encoding
\usepackage[T1]{fontenc}		%Font encoding
\usepackage[english=american]{csquotes}	%Quotation marks ' '
\usepackage{graphicx}			%Include graphics
\usepackage{epstopdf}
\usepackage{subfig}				%Subfigures
\usepackage{marvosym}			%Special symbols
\usepackage{amsmath}			%Mathmode
\usepackage{amssymb}
\usepackage{amsthm}
\usepackage{mathtools}
\usepackage{color}

\theoremstyle{definition}
\newtheorem{definition}{Definition}%[subsection]
\newtheorem{example}{Example}%[subsection]

\newtheorem*{remark}{Remark}
\newtheorem*{remarks}{Remarks}
\theoremstyle{plain}

\newtheorem{theorem}{Theorem}%[section]
\newtheorem{proposition}{Proposition}

\DeclareMathOperator{\TN}{\mathsf{T}(\mathcal{N})}
\DeclareMathOperator*{\argmax}{argmax}

\begin{document}

\begin{frontmatter}
\title{Phylogenetic diversity and biodiversity indices on phylogenetic networks}

\author{Kristina Wicke}
\ead{kristina.wicke@uni-greifswald.de}

\author{Mareike Fischer \corref{cor1}}
\ead{email@mareikefischer.de}
\cortext[cor1]{Corresponding author}

\address{Institute of Mathematics and Computer Science, University of Greifswald, Germany}

\begin{abstract}
In biodiversity conservation it is often necessary to prioritize the species to conserve. 
Existing approaches to prioritization, e.g. the Fair Proportion Index and the Shapley Value, are based on phylogenetic trees and rank species according to their contribution to overall phylogenetic diversity.
However, in many cases evolution is not treelike and thus, phylogenetic networks have been developed as a generalization of phylogenetic trees, allowing for the representation of non-treelike evolutionary events, such as hybridization.
Here, we extend the concepts of phylogenetic diversity and phylogenetic diversity indices from phylogenetic trees to phylogenetic networks. 
On the one hand, we consider the treelike content of a phylogenetic network, e.g. the (multi)set of phylogenetic trees displayed by a network and the so-called lowest stable ancestor tree associated with it. On the other hand, we derive the phylogenetic diversity of subsets of taxa and biodiversity indices directly from the internal structure of the network. 
We consider both approaches that are independent of so-called inheritance probabilities as well as approaches that explicitly incorporate these probabilities. 
Furthermore, we introduce our software package NetDiversity, which is implemented in Perl and allows for the calculation of all generalized measures of phylogenetic diversity and generalized phylogenetic diversity indices established in this note that are independent of inheritance probabilities. 
We apply our methods to a phylogenetic network representing the evolutionary relationships among swordtails and platyfishes (\textit{Xiphophorus}: Poeciliidae), a group of species characterized by widespread hybridization.
\end{abstract}

\begin{keyword}
Hybridization \sep Phylogenetic networks \sep Phylogenetic diversity \sep Shapley Value \sep Fair Proportion Index
\end{keyword}
\end{frontmatter}

%\linenumbers

\section{Introduction}
Facing a major extinction crisis and the inevitable loss of biodiversity at the same time with limited financial means, biological conservation has to prioritize the species to conserve. In this matter, the so-called phylogenetic diversity (\citet{Faith1992}) has been introduced as a measure of biodiversity based on the evolutionary history of species. It serves as a basis for biodiversity indices used in taxon prioritization, e.g. the Fair Proportion Index and the Shapley Value (\citet{Haake2007, Hartmann2013, Fuchs2015, Wicke2015}). \\
Both phylogenetic diversity, as well as the Fair Proportion Index and the Shapley Value are based on phylogenetic trees and thus, assume the evolutionary history of species to be treelike.
However, there are several forms of non-treelike evolution, such as hybridization, affecting a variety of species. 
Therefore, phylogenetic reticulation networks have become an important concept in evolutionary biology, allowing for the representation of non-treelike evolution. \\
Here, we aim at combining both approaches, i.e. we aim at extending the concept of phylogenetic diversity and its measures from phylogenetic trees to phylogenetic networks. So far, phylogenetic diversity and the Shapley Value have been considered for so-called split networks, which can be used to represent conflict in data (\citet{Chernomor2016, Volkmann2014}), but no attempts have been made towards the generalization of phylogenetic diversity and its measures to reticulation networks. \\
In this note we first recapitulate phylogenetic diversity, the Fair Proportion Index and the Shapley Value on phylogenetic trees, before we focus on generalizing these concepts to phylogenetic networks. \\
We will introduce a variety of definitions for generalized phylogenetic diversity, following three main principles: the calculation of spanning arborescences and subgraphs of a network, the consideration of the (multi)set of phylogenetic trees displayed by a network and the construction of the so-called lowest stable ancestor tree associated with a network. \\
We will then turn our attention to the Fair Proportion Index and the Shapley Value and suggest different ways of using them as taxon prioritization tools in the context of phylogenetic networks. \\
Both for the generalized measures of phylogenetic diversity and the generalized biodiversity indices, we develop both approaches that are independent of so-called inheritance probabilities as well as approaches that explicitly incorporate these probabilities. \\
In case of the former, all approaches are implemented in our new software tool NetDiversity, which has been made publicly available at \\
www.mareikefischer.de/Software/NetDiversity.zip. \\
Moreover, we test NetDiversity on a recently published phylogenetic network of swordtails and platyfishes (\textit{Xiphophorus}: Poeciliidae), whose evolution is characterized by widespread hybridization (\citet{Solis-Lemus2016}).

\section{Preliminaries}
%Definition phylogenetic tree
Let $X$ be a finite set of species (taxa). 
A \emph{rooted phylogenetic $X$-tree $\mathcal{T}$} is a rooted tree with root $\rho$ where the leaves are bijectively labeled by $X$. $\mathcal{T}$ is called \emph{binary} if all internal nodes have degree $3$ and the root has degree $2$. 
Throughout this paper, when we refer to trees, we always mean rooted phylogenetic trees. 
Furthermore, we assume all edges in a tree to have edge lengths greater than zero assigned to them, and we denote the length of an edge $e$ as $\lambda_e > 0$. \\
Note that all edges in a rooted phylogenetic tree $\mathcal{T}$ are directed away from the root, thus formally the treeshape of $\mathcal{T}$ is a so-called \emph{arborescence}.

\begin{definition}[Arborescence] \label{def_arborescence}
Let $G=(V,E)$ be a directed graph and let $\rho \in V$ be a specified root node (of indegree 0). Then $G$ is an \emph{arborescence} (rooted at $\rho$) if there is exactly one directed path from $\rho$ to $u$ for all nodes $u \in V \setminus \{\rho\}$.
\end{definition}

%Definition phylogenetic network
\noindent
A \emph{rooted binary phylogenetic network} $\mathcal{N}$ on $X$ is a connected rooted acyclic digraph such that:
	\begin{itemize}
	\item the root has outdegree 2 (and indegree 0),
	\item each node with outdegree 0 has indegree 1, and the set of nodes with outdegree 0 is bijectively labeled by $X$,
	\item all other nodes either have indegree 1 and outdegree 2, or indegree 2 and outdegree 1. 
	\end{itemize}
Nodes with indegree 2 and outdegree 1 are called \emph{reticulation nodes} and all other nodes are called \emph{tree nodes}. Furthermore, tree nodes with outdegree 0 are referred to as \emph{leaves}.
Edges directed into a reticulation node are called \emph{reticulation edges} and edges directed into a tree node are called \emph{tree edges}. 
When we refer to phylogenetic networks, we always mean rooted binary phylogenetic networks.
Moreover, when we refer to the size of a tree or a network, we mean the number $n = \vert X \vert$ of taxa, i.e. the number of leaves of the tree or network under consideration. \\

\noindent Additionally, we assume all tree edges of a phylogenetic network to have edge lengths greater than zero assigned to them and denote the length of a tree edge $e$ as $\lambda_e > 0$. W.l.o.g. we define the edge lengths of all reticulation edges to be zero. 
However, we assign so-called \emph{inheritance probabilities} to the reticulation edges of a network, reflecting the probability with which a hybrid species inherits its genetic material from both of its parents. More formally, let $\mathcal{N}$ be a phylogenetic network on $X$ and let $r$ be a reticulation node, i.e. a hybrid species, with parents $p_1$ and $p_2$. 
Let $e_1=(p_1, r)$ be the edge between $p_1$ and $r$ and analogously let $e_2=(p_2,r)$ be the edge between $p_2$ and $r$. Then we use $\gamma_{e_1} \in (0,1)$ to denote the probability that $r$ inherits its genetic material (e.g. a nucleotide or a gene) from $p_1$ and we use $\gamma_{e_2} = 1 - \gamma_{e_1}$ to denote the probability that the genetic material is inherited from $p_2$. We call $\gamma_{e_1}$ and $\gamma_{e_2}$ \emph{inheritance probabilities} and associate $\gamma_{e_1}$ with edge $e_1$ and $\gamma_{e_2}$ with $\gamma_{e_2}$ (cf. Figure \ref{fig_embedded}). 
If no inheritance probabilities are given, we assume $\gamma_{e_{1_i}} = \gamma_{e_{2_i}} = \frac{1}{2}$ for all reticulation nodes $r_i$. Moreover, we assign probability one to all tree edges, i.e. the probability $\mathbb{P}(e)$ assigned to an edge $e$ is given by
$$ \mathbb{P}(e) = \begin{cases}
				1, & \text{if $e$ is a tree edge of $\mathcal{N}$;} \\
				\frac{1}{2}, & \parbox[t]{.6\textwidth}{if $e$ is a reticulation edge of $\mathcal{N}$, but does not have an inheritance probability assigned to it;} \\
				\gamma_e, & \parbox[t]{.6\textwidth}{if $e$ is a reticulation edge of $\mathcal{N}$ and $\gamma_e \in (0,1)$ is the inheritance probability assigned to $e$.}
				\end{cases}$$

%Definition embedded trees
\noindent
Let $\mathcal{N}$ be a phylogenetic network on $X$ and let $\mathcal{T}$ be a phylogenetic $X$-tree.
We say that $\mathcal{T}$ is embedded in $\mathcal{N}$, or that $\mathcal{N}$ displays $\mathcal{T}$, if $\mathcal{T}$ can be obtained from $\mathcal{N}$ by deleting one of the reticulation edges for each reticulation node and suppressing resulting nodes of indegree 1 and outdegree 1.
We use $\TN$ to denote the (multi)set of all rooted phylogenetic $X$-trees displayed by $\mathcal{N}$. \\
Note that we receive the edge weights of an embedded tree $\mathcal{T} \in \TN$ as follows: for all formerly distinct edges that are melted into a new edge by suppressing nodes of indegree 1 and outdegree 1, we add their edge lengths, while all other edges keep their original weights.
Moreover, note that if there are $k$ reticulation nodes in a rooted binary phylogenetic network $\mathcal{N}$ on a taxon set $X$, then there are at most $2^k$ phylogenetic $X$-trees displayed by $\mathcal{N}$. However, this bound does not have to be sharp (cf. Figure \ref{fig_embedded}). \\
In the following we will also need the probability $\mathbb{P}(\mathcal{T})$ of an embedded tree, which is calculated as follows: 
	\begin{enumerate}
	\item For all $\mathcal{T} \in \TN$ calculate the unscaled probability
		$$ \mathbb{P}_{\text{unscaled}}(\mathcal{T}) = \prod\limits_{\substack{e: \, e \text{ is reticulation edge} \\ \text{and is kept in constructing } \mathcal{T}}} \gamma_e,$$
		where $\gamma_e$ is the inheritance probability associated with $e$.
	\item Set $p \coloneqq \sum\limits_{\mathcal{T} \in \TN} \mathbb{P}_{\text{unscaled}}(\mathcal{T})$ (scaling factor).
	\item Calculate the probability 
	$$ \mathbb{P}(\mathcal{T}) = \frac{1}{p} \cdot \mathbb{P}_{\text{unscaled}}(\mathcal{T}).$$
	\end{enumerate}
Here, the scaling factor $p$ ensures that the probabilities of all embedded trees sum up to one.
\begin{example}\label{prob_tree}
Consider the phylogenetic network $\mathcal{N}$ on $X=\{A,B,C,D\}$ and its embedded trees $\mathcal{T}_1, \, \mathcal{T}_2$ and $\mathcal{T}_3$. We have $\mathbb{P}_{\text{unscaled}}(\mathcal{T}_1) = \frac{2}{3} \cdot \frac{3}{4} = \frac{1}{2}$ and analogously $\mathbb{P}_{\text{unscaled}}(\mathcal{T}_2) = \frac{1}{6}$ and $\mathbb{P}_{\text{unscaled}}(\mathcal{T}_3) = \frac{1}{4}$. Thus, $p = \frac{1}{2} + \frac{1}{6} + \frac{1}{4}$ and we retrieve the following probabilities of the embedded trees:  $\mathbb{P}(\mathcal{T}_1) = \frac{12}{11} \cdot \frac{1}{2} = \frac{6}{11}$ and analogously $\mathbb{P}(\mathcal{T}_2) = \frac{2}{11}$ and $\mathbb{P}(\mathcal{T}_3) = \frac{3}{11}$.
\end{example}

%Figure embedded trees
\begin{figure}[htbp]
	\centering
	\includegraphics[scale=0.6]{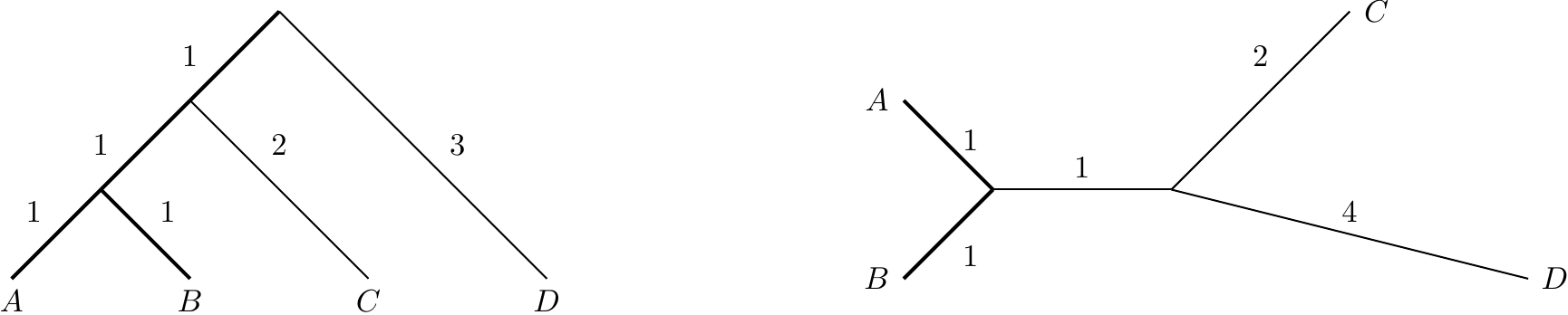}
	\caption{Rooted phylogenetic network $\mathcal{N}$ on $X=\{A,B,C,D\}$ with edge lengths associated with the tree edges (solid) and inheritance probabilities associated with the reticulation edges (dashed).  $\mathcal{N}$ displays the phylogenetic $X$-trees $\mathcal{T}_1, \mathcal{T}_2$ and $\mathcal{T}_3$. When deleting exactly one reticulation edge for each of the two reticulation nodes $r_1$ and $r_2$ in $\mathcal{N}$, we also obtain tree $\mathcal{T}_4$, in which the internal node $w$ of $\mathcal{N}$ has become a leaf. However, we do not regard $\mathcal{T}_4$ as a phylogenetic $X$-tree displayed by $\mathcal{N}$, because $w$ does not belong to taxon set $X$. Thus, in this case we have $\TN = \{\mathcal{T}_1, \mathcal{T}_2, \mathcal{T}_3\}$. \\
	The bold edges in $\mathcal{T}_1$ represent the edges contributing to the phylogenetic diversity of $S=\{A,B\}$ calculated in Example \ref{example_pd}.}
	\label{fig_embedded}
\end{figure}

%Definition LSA tree
\noindent
For a phylogenetic network $\mathcal{N}$ and a node $u$ of $\mathcal{N}$ that is not the root, we call  any node $v$ that lies on all directed paths from the root to $u$ a \emph{stable ancestor} of $u$. The so-called \emph{lowest stable ancestor} of $u$ is defined as the last node $lsa(u)$ that is contained on all paths from the root to $u$, excluding $u$. 
Based on this terminology we can define the \emph{lowest stable ancestor tree or \emph{LSA tree} (cf. \citet{Huson2011}, p. 140)} associated with a network.
Let $\mathcal{N}$ be rooted phylogenetic network on $X$. The LSA tree $\mathcal{T}_{LSA}(\mathcal{N})$ associated with $\mathcal{N}$ is a rooted phylogenetic $X$-tree that can be computed as follows: 
For each reticulation node $r$ in $\mathcal{N}$, remove all edges directed into $r$ and add a new edge $e=(lsa(r),r)$ from the lowest stable ancestor of $r$ into $r$. Then repeatedly remove all unlabeled leaves and nodes with in- and outdegree 1, until no further such removal is possible.
Note that the LSA tree associated with a binary rooted phylogenetic network is not necessarily a binary phylogenetic tree (cf. Figure \ref{fig_lsatree}). 
Note that every node $v$ in a phylogenetic network $\mathcal{N}$ has a unique lowest stable ancestor $lsa(v)$. Thus, the LSA tree associated with a given network is the same regardless of the order that the reticulation nodes are processed in. Moreover, note that the concept of a lowest stable ancestor is not new, but has long been used in the theory of flow graphs, where the lowest stable ancestor $lsa(v)$ of a node $v$ is called the \emph{immediate dominator} of $v$ and the LSA tree is called the \emph{dominator tree} of the flow graph (cf. \citet{Lengauer1979}).
\\
In order to use the LSA tree for subsequent phylogenetic diversity calculations, we have to infer edge lengths for the edges of the LSA tree. For all tree edges of $\mathcal{N}$ that are also present in $\mathcal{T}_{LSA}(\mathcal{N})$, we use their original edge weights. If during the removal of nodes of in-and outdegree 1 two formerly distinct tree edges of $\mathcal{N}$ are melted into a new edge in $\mathcal{T}_{LSA}(\mathcal{N})$, we add their original edge lengths. 
For all newly established edges $e=(lsa(r),r)$ between a reticulation node $r$ and its lowest stable ancestor, we suggest to set the length of these edges to the average path length of a path between $lsa(r)$ and $r$, respectively, i.e. we set
$$ \lambda_{e=(lsa(r),r)} \coloneqq \frac{1}{\vert \mathcal{P}_r \vert} \sum\limits_{P \in \mathcal{P}_r} length(P),$$
where $\mathcal{P}_r$ is the set of all $lsa(r)$-$r$-paths $P$ in $\mathcal{N}$ and the length of any such path is obtained by adding the edge lengths of all edges that are part of this path (cf. Figure \ref{fig_lsatree}).

\begin{remark}
Note that instead of using the average path length between a reticulation node $r$ and its lowest stable ancestor $lsa(r)$ in order to infer a weight for the edge $e = (lsa(r),r)$, we could also use the length of a shortest path, the length of a most likely path or a weighted average path length, where each path $P$ is weighted according to its probability 
$$ \mathbb{P}(P) = \prod\limits_{e: \, e \text{ is edge of } P} \mathbb{P}(e).$$
\end{remark}

%Figure LSA tree
\begin{figure}[htbp]
	\centering
	\includegraphics[scale=0.6]{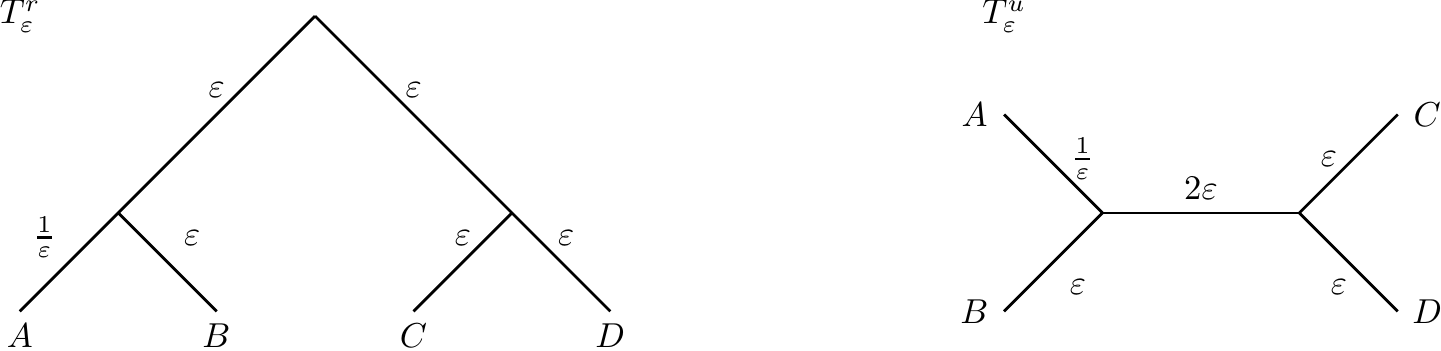}
	\caption{Rooted binary phylogenetic network $\mathcal{N}$ on $X=\{A,B,C,D\}$ and its associated LSA tree $\mathcal{T}_{LSA}(\mathcal{N})$. Note that the reticulation edges (dashed) of $\mathcal{N}$ have weight zero. The node $v$ is the lowest stable ancestor of the reticulation node $r_1$ and we have to consider two paths when calculating the length of the edge $e=(lsa(r_1),r_1)$: $P_1 = ((v,u),(u,r_1))$ with $length(P_1)=1+0=1$ (recall that we have defined the lengths of reticulation edges to be zero) and $P_2 = ((v,w),(w,r_1))$ with length $length(P_2)=1+0=1$. Thus, taking the average, we set $length((lsa(r_1),r_1)) \coloneqq 1$. Analogously, node $\rho$ is the lowest stable ancestor of $r_2$ and we have to consider the paths $P_3=((\rho,v)(v,w)(w,r_2))$ with $length(P_3)=1+1+0=2$ and $P_4=((\rho,x),(x,r_2))$ with $length(P_4)=2+0=2$. Thus, we set $length((lsa(r_2),r_2)) \coloneqq 2$. However, subsequently the edges $(v,r_1)$ and $(r_1,B)$ are merged into a new edge $(v,B)$ of length $1+1=2$ and analogously, the edges $(\rho,r_2)$ and $(r_2,C)$ are replaced by a new edge $(\rho,C)$ of length $2+1=3$ to finally yield the LSA tree associated with $\mathcal{N}$. Note that $\mathcal{T}_{LSA}(\mathcal{N})$ is not binary, because the root $\rho$ has degree $3$. } 
	\label{fig_lsatree}
\end{figure}

\subsection{Phylogenetic diversity and phylogenetic diversity indices on trees}
In this section we briefly recapitulate the concept of phylogenetic diversity and phylogenetic diversity indices, in particular the Shapley Value and the Fair Proportion Index, for phylogenetic trees.

%Definition phylogenetic diversity
\begin{definition}[Phylogenetic diversity] \label{def_pd}
Let $\mathcal{T}$ be a rooted phylogenetic tree with leaf set $X$.
For a subset $S \subseteq X$ of taxa, the \emph{phylogenetic diversity} $PD(S)$ is calculated by summing up the edge lengths of the phylogenetic subtree of $\mathcal{T}$ containing $S$ and the root (i.e., we consider the sum of edge lengths in the smallest spanning tree containing $S$ and the root).
\end{definition}

%Example phylogenetic diversity
\begin{example} \label{example_pd}
Consider the phylogenetic tree $\mathcal{T}_1$ on $X=\{A,B,C,D\}$ depicted in Figure \ref{fig_embedded}. Now consider the subset $S=\{A,B\} \subseteq X$ of taxa. Then the \emph{phylogenetic diversity} of $S$ calculates as $PD(S) = 2 + 1 + 1 + 1 = 5$.
\end{example}

\noindent
Based on phylogenetic diversity, we can now define the Shapley Value for phylogenetic trees. The Shapley Value for phylogenetic trees is used in different versions in the literature (cf. \citet{Wicke2015}), but we will use the so-called original Shapley Value throughout this paper.

%Definition original Shapley Value
\begin{definition}[Original Shapley Value] \label{def_sv}
Let $\mathcal{T}$ be a rooted phylogenetic tree with leaf set $X$ and let $PD(S)$ denote the phylogenetic diversity of $S \subseteq X$. Then the Shapley Value for a taxon $a \in X$ is defined as
\begin{equation}
	SV_{\mathcal{T}}(a) = \frac{1}{n!}\sum_{\substack{S \subseteq X \\ a \in S}}(\lvert S \rvert -1)!(n- \lvert S \rvert)!(PD(S)-PD(S \setminus \{a\})),
\end{equation}
where $n = \lvert X \rvert$ and $S$ denotes a subset of species containing taxon $a$ (also sometimes referred to as `coalition') and the sum runs over all such coalitions possible. 
\end{definition}

\noindent
While the Shapley Value reflects the average contribution of a species to overall phylogenetic diversity and is thus a sensible prioritization criterion, its calculation is complicated. 
Therefore another index, the so-called Fair Proportion Index, has been introduced. 

%Definition Fair Proportion Index
\begin{definition}[Fair Proportion Index] \label{def_fp}
For a rooted phylogenetic tree $\mathcal{T}$ with leaf set $X$ the Fair Proportion Index of a taxon $a$ is defined as
	\begin{equation} 
	FP_{\mathcal{T}}(a) = \sum_{e} \frac{\lambda_{e}}{D_{e}},
	\end{equation}
where the sum runs over all edges $e$ on the path from $a$ to the root and $D_e$ denotes the number of leaves descendent from that edge.
\end{definition}

\noindent
The Fair Proportion Index can easily be calculated, but lacks a biological motivation. However, its use has been justified by its equivalence with the original Shapley Value.

\begin{theorem}[\citet{Fuchs2015}] \label{fuchs_theorem}
Let $\mathcal{T}$ be a rooted phylogenetic tree with leaf set $X$. Then we have for all $a \in X:$
$$SV_{\mathcal{T}}(a) = FP_{\mathcal{T}}(a).$$
\end{theorem}
%Example Fair Proportion Index
\begin{example}
Consider the phylogenetic tree $\mathcal{T}_1$ on $X=\{A,B,C,D\}$ depicted in Figure \ref{fig_embedded}. Here, we have $FP_{\mathcal{T}_1}(A) = \frac{1}{3} + \frac{2}{1} = \frac{7}{3}, \, FP_{\mathcal{T}_1}(B) = \frac{1}{3} + \frac{1}{2} + \frac{1}{1} = \frac{11}{6}, \, FP_{\mathcal{T}_1}(C) = \frac{1}{3} + \frac{1}{2} + \frac{1}{1} = \frac{11}{6}$ and $ FP_{\mathcal{T}_1}(D) = \frac{3}{1} = 3$. Note that $FP_{\mathcal{T}_1}(A) + FP_{\mathcal{T}_1}(B) + FP_{\mathcal{T}_1}(C) + FP_{\mathcal{T}_1}(D) = 9$, which equals the total sum of all edge lengths in $\mathcal{T}_1$. Also note that the Fair Proportion Indices of $\mathcal{T}_1$ equal the Shapley Values of $\mathcal{T}_1$.
\end{example}

\section{Generalization of phylogenetic diversity}
We are now in the position to present our approaches towards the generalization of phylogenetic diversity from trees to networks. We will introduce three approaches, one based on the calculation of spanning arborescences and subgraphs of a network, one based on the set of trees displayed by a network and one based on the LSA tree associated with a network.

\subsection{Phylogenetic (sub)net diversity}
Recall that the phylogenetic diversity of a subset $S \subseteq X$ of taxa of a phylogenetic $X$-tree $\mathcal{T}$ was calculated as the sum of branch lengths of the subtree of $\mathcal{T}$ containing $S$ and the root. For a phylogenetic network $\mathcal{N}$ on $X$ and a subset $S \subseteq X$ of taxa, there may be more than one subtree, or to be precise, more than one arborescence (because a phylogenetic network is a directed graph) containing $S$ and the root. Thus, we suggest to consider an arborescence of minimum cost, i.e. an arborescence whose weight (the sum of its branch lengths) is no larger than the weight of any other arborescence spanning $S$ and the root, and introduce the so-called \emph{phylogenetic net diversity}.

%Definition phylogenetic net diversity
\begin{definition}[Phylogenetic net diversity]  \label{def_pnd}
Let $\mathcal{N}$ be a rooted phylogenetic network on some taxon set $X$. 
For a subset $S \subseteq X$ of taxa we define the \emph{phylogenetic net diversity} $PND(S)$ of $S$ as the sum of branch lengths in a minimum cost arborescence containing $S$ and the root.
\end{definition}

\noindent
Note that determining the minimum cost arborescence containing a subset $S \subseteq X$ of taxa and the root is formally an instance of the so-called \emph{directed Steiner tree problem} or \emph{Steiner arborescence problem}, which, in general, is an $NP$-hard problem (\citet{Karp1972}). \\

\noindent 
In order to explicitly incorporate the inheritance probabilities of a network into the calculation of phylogenetic net diversity, several alterations of Definition \ref{def_pnd} are possible. Instead of considering a minimum cost arborescence spanning the taxa in $S$ and the root, we could consider all arborescences spanning $S$ and the root and weight them according to their probability or use a most likely arborescence. We denote these values by $PND^{inh}$ and $PND^{ML}$, i.e.
$$ PND^{inh}(S)  = \sum\limits_{A \in \mathcal{A}_S} \mathbb{P}(A) \cdot weight(A),$$
where $\mathcal{A}_S$ denotes the set of all arborescences spanning $S$ and the root, $weight(A)$ is the sum of branch lengths of any such arborescence $A$ and 
$$\mathbb{P}(A) = \prod\limits_{e: \, e \text{ is edge of } A} \mathbb{P}(e)$$
denotes its probability. Moreover, 
$$ PND^{ML} = weight(A'), \text{ where } A' = \argmax\limits_{A \in \mathcal{A}_S} \mathbb{P}(A).$$
If the argmax is not unique, we choose one of the most likely arborescences of minimum cost.\footnote{Alternatively, we could arbitrarily choose one of the most likely arborescences. However, choosing an arborescence of minimum cost makes the results reproducible.}

%Figure Network and Arborescences
\begin{figure}[htbp]
	\centering
	\includegraphics[scale=0.6]{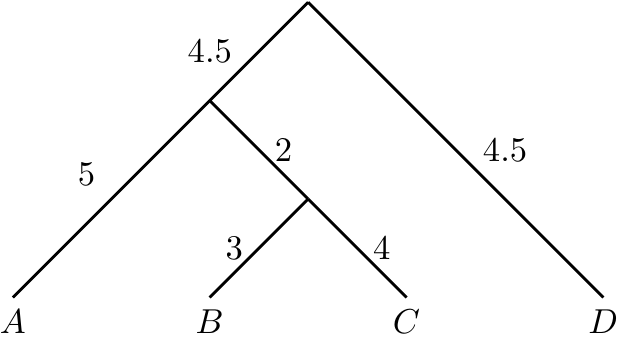}
	\caption{Rooted binary phylogenetic network $\mathcal{N}$ on $X = \{A,B,C,D\}$ and arborescences $A_1$ and $A_2$ containing $S=\{A,B\}$ and the root. \\
	The bold edges in $\mathcal{N}$ depict the subgraph $N_{\{A,B\}}$ of $\mathcal{N}$ containing $\rho$ and $S=\{A,B\}$. Note that all reticulation edges (dashed) have weight zero.}
	\label{fig_arborescences}
\end{figure}

%Example phylogenetic net diversity
\begin{example}
Consider Figure \ref{fig_arborescences}, which depicts the rooted phylogenetic network $\mathcal{N}$ on $X=\{A,B,C,D\}$ and the two arborescences $A_1$ and $A_2$ containing $S=\{A,B\}$ and the root. $A_1$ has weight $1+1+2=4$, while $A_2$ has weight $2+1+1+1=5$. Thus, $A_1$ is the \emph{minimum cost arborescence} containing $S=\{A,B\}$ and the root and we retrieve the \emph{phylogenetic net diversity} of $S = \{A,B\}$ as $PND(\{A,B\}) = 4$. 
However, $A_1$ has probability $\frac{1}{3}$ and $A_2$ has probability $\frac{2}{3}$, i.e. $A_2$ is the most likely arborescence spanning $S$ and the root. Thus, $PND^{ML}(\{A,B\}) = weight(A_2) = 5$. Moreover, $PND^{inh}(\{A,B\}) = \frac{1}{3} \cdot 4 + \frac{2}{3} \cdot 5 = \frac{14}{3}$.
\end{example}

\noindent Instead of using spanning arborescences to define the phylogenetic diversity of a subset $S \subseteq X$ of taxa of a phylogenetic network $\mathcal{N}$ on $X$, we can also consider the subgraph $\mathcal{N}_S \subseteq \mathcal{N}$ containing the root of $\mathcal{N}$ and $S$ and define the phylogenetic diversity of $S$ as the sum of branch lengths in $\mathcal{N}_S$.

%Definition of phylogentic subnet diversity
\begin{definition}[Phylogenetic subnet diversity]
Let $\mathcal{N}$ be a rooted phylogenetic network on some taxon set $X$. 
For a subset $S \subseteq X$ of taxa consider the subgraph $\mathcal{N}_S$ of $\mathcal{N}$ containing the root of $\mathcal{N}$ and the taxa in $S$ (i.e., $\mathcal{N}_S$ is the subgraph of $\mathcal{N}$ containing all nodes and edges that lie on at least one path from the root of $\mathcal{N}$ to any of the leaves in $S$). Then we define the \emph{phylogenetic subnet diversity} $PSD(S)$ of $S$ as the sum of branch lengths in $\mathcal{N}_S$.
\end{definition}

%Example phylogenetic subnet diversity
\begin{example}
Consider the rooted phylogenetic network $\mathcal{N}$ on $X=\{A,B,C,D\}$ depicted in Figure \ref{fig_arborescences} and set $S=\{A,B\}$. Then the subgraph $\mathcal{N}_S$ of $\mathcal{N}$ (highlighted with bold lines) has length $1+1+1+1+1=5$ and thus, $PSD(\{A,B\}) = 5$.
\end{example}

\subsection{Embedded phylogenetic diversity}
If species are subject to hybridization or horizontal gene transfer, their genome contains parts of the genome of both its ancestors. However, evolution at the nucleotide level rather than the genome level is still treelike, because a single nucleotide can always be traced back to one parent. 
Therefore, we suggest to consider the set of trees embedded in a network as an alternative approach towards the generalization of phylogenetic diversity from trees to networks.

%Definition embedded phylogenetic diversity
\begin{definition}[Embedded phylogenetic diversity] \label{def_embedded_pd}
Let $\mathcal{N}$ be a rooted phylogenetic network on some taxon set $X$ and let $\TN$ be the (multi)set of all rooted phylogenetic $X$-trees displayed by $\mathcal{N}$. 
Then we use $PD_{\TN}^{\ast}(S)$ to denote the \emph{embedded phylogenetic diversity} of a subset $S \subseteq X$ of taxa, where $\ast$ is one of the following functions $\min, \max, \sum, \varnothing$
and define 
\begin{align}
PD_{\TN}^{\min}(S) &\coloneqq \min_{\mathcal{T} \in \TN} \{ PD_{\mathcal{T}}(S) \}, \\
PD_{\TN}^{\max}(S) &\coloneqq \max_{\mathcal{T} \in \TN} \{ PD_{\mathcal{T}}(S) \}, \\
PD_{\TN}^{\sum} (S) &\coloneqq \sum_{\mathcal{T} \in \TN} PD_{\mathcal{T}}(S) \text{ and } \\
PD_{\TN}^{\varnothing}(S) &\coloneqq \frac{1}{\vert \TN \vert} \sum_{\mathcal{T} \in \TN} PD_{\mathcal{T}}(S), \label{average_embedded_pd}
\end{align}
where $\vert \TN \vert$ is the number of phylogenetic $X$-trees displayed by $\mathcal{N}$.
If inheritance probabilities are given for $\mathcal{N}$, we also consider
\begin{align}
PD_{\TN}^{\varnothing_{inh}}(S) &\coloneqq  \sum_{\mathcal{T} \in \TN} \mathbb{P}(\mathcal{T}) \cdot PD_{\mathcal{T}}(S) \text{ and } \\
PD_{\TN}^{ML}(S) &\coloneqq PD_{\mathcal{T}'}(S) \text{ with } \mathcal{T}' = \argmax_{\mathcal{T} \in \TN} \mathbb{P}(\mathcal{T}),
\end{align}
where $\mathbb{P}(\mathcal{T})$ is the probability of $\mathcal{T}$ and $\mathcal{T}'$ is a most likely embedded tree. 
If the argmax is not unique, we arbitrarily choose one of the embedded trees with maximum probability.
\end{definition}

\noindent
Note that $\ast$ can be replaced by other functions on the phylogenetic diversity of the trees in $\TN$, but we will only consider $\min, \max, \sum, \varnothing$ and $\varnothing_{inh}$ as defined above. \\
Also note that we will only consider phylogenetic $X$-trees as elements of $\TN$ and discard all other trees that may occur when decomposing the network into a set of trees (cf. Figure \ref{fig_embedded}).

%Example embedded phylogenetic diversity
\begin{example}
Consider the rooted phylogenetic network $\mathcal{N}$ on $X=\{A,B,C,D\}$ and its embedded trees $\mathcal{T}_1, \mathcal{T}_2$ and $\mathcal{T}_3$ depicted in Figure \ref{fig_embedded}. Now set $S=\{A,B\} \subseteq X$. Then we have $PD_{\mathcal{T}_1}(S) = 5, PD_{\mathcal{T}_2}(S) = 5$ and $PD_{\mathcal{T}_3}(S) = 4$.
Moreover, $\mathbb{P}(\mathcal{T}_1) = \frac{6}{11}, \, \mathbb{P}(\mathcal{T}_2) = \frac{2}{11}$ and $\mathbb{P}(\mathcal{T}_3) = \frac{3}{11}$.
Thus, we retrieve the different values of the \emph{embedded phylogenetic diversity} of $S=\{A,B\}$ as
$PD_{\TN}^{\min}(S) = 4, \, PD_{\TN}^{\max}(S) = 5, \, PD_{\TN}^{\sum}(S) = 14, \, PD_{\TN}^{\varnothing}(S) = \frac{14}{3}, \, PD_{\TN}^{\varnothing_{inh}}(S) = \frac{52}{11}$ and $PD_{\TN}^{ML}(S) = 5$.
\end{example}

%Relationship PND and PD_min
\subsection{Relationship between the phylogenetic net diversity and the embedded phylogenetic diversity}
\noindent
Comparing the phylogenetic net diversity $PND$ and the minimum embedded phylogenetic diversity $PD_{\TN}^{\min}$ for a subset $S \subseteq X$ of taxa, we see that they use a similar principle. While $PND(S)$ is defined as the weight of a minimum cost arborescence spanning $S$ and the root in a network $\mathcal{N}$, $PD_{\TN}^{\min}$ is defined as the weight of a minimum spanning tree/minimum cost arborescence spanning $S$ and the root in the set $\TN$ of phylogenetic $X$-trees displayed by $\mathcal{N}$. 
Thus, the two measures are related, but in general they are not identical.
Consider, for example the rooted phylogenetic network $\mathcal{N}$ depicted in Figure \ref{fig_embedded} and set $S=\{A,B,C,D\}$. Then, we have $PD_{\TN}^{\min}(S) = 9$, while $PND(S) = 8$. \\
However, we have the following relationship between $PND$ and $PD_{\TN}^{\min}$:

%Proposition
\begin{proposition}\label{PND_equality}
Let $\mathcal{N}$ be a binary rooted phylogenetic network on a taxon set $X$ with $k$ reticulation nodes and let $\TN$ be the set of phylogenetic $X$-trees displayed by $\mathcal{N}$.
	\begin{enumerate}
	\item We have
		\begin{equation} \label{leq}
		PND(S) \leq  PD_{\mathsf{T}(\mathcal{N})}^{\min}(S)
		\end{equation}
		for all subsets $S \subseteq X$ of taxa.
	\item If $\vert \TN \vert = 2^{k}$, i.e. if all combinations of removing one reticulation edge for each reticulation node and suppressing nodes of both indegree 1 and outdegree 1 result in a phylogenetic $X$-tree, we have
		\begin{equation}
		PND(S) =  PD_{\mathsf{T}(\mathcal{N})}^{\min}(S).
		\label{eq}
		\end{equation}
	\end{enumerate}
\end{proposition}

%Remark
\begin{remark}
Note that $\vert \TN \vert = 2^k$ for example holds for so-called \emph{normal} networks (cf. \citet{Iersel2010}).
\end{remark}

%Proof
\begin{proof}[Proof of Proposition \ref{PND_equality}]\label{equality_proof}
Let $\mathcal{N}$ be a binary rooted phylogenetic network with root $\rho$, taxon set $X$ and $k$ reticulation nodes.
Let $\TN$ be the set of embedded trees and let $R(\mathcal{N}) = \{r \, \vert \, r \text{ is a reticulation node of } \mathcal{N}\}$ be the set of reticulation nodes of $N$. 
	\begin{enumerate}
	\item We show $PD_{\mathsf{T}(\mathcal{N})}^{\min}(S) \geq PND(S)$. \\
		  For every $\mathcal{T} \in \TN$ the phylogenetic diversity of a subset $S \subseteq X$ of taxa is defined as the sum of branch lengths in the smallest arborescence spanning the taxa in $S$ and the root. 
Clearly, the weight of any such arborescence cannot be smaller than the weight of a minimum cost arborescence spanning $S$ and the root in $\mathcal{N}$ (all $\mathcal{T} \in \TN$ are \enquote{subgraphs} of $\mathcal{N}$, thus, any smallest arborescence spanning $S$ and the root in a displayed tree $\mathcal{T} \in \TN$ can also be found in $\mathcal{N}$).\footnote{Formally, we have to re-establish the nodes of in- and outdegree 1 that were removed during the construction of $\mathcal{T} \in \TN$ to make $\mathcal{T}$ a subgraph of $\mathcal{N}$. However, this does not affect the weights.} In particular, we have
$$ \min\limits_{\mathcal{T} \in \TN} \{ PD_{\mathcal{T}}(S) \} = PD_{\mathsf{T}(\mathcal{N})}^{\min}(S) \geq PND(S).$$ 

	\item Now, suppose that $\vert \TN \vert = 2^{k}$. 
		  We want to show that $PND(S) =  PD_{\mathsf{T}(\mathcal{N})}^{\min}(S)$.
		  As we have $PND(S) \leq PD_{\mathsf{T}(\mathcal{N})}^{\min}(S)$ (Equation \eqref{leq}), it suffices to show $PND(S) \geq PD_{\mathsf{T}(\mathcal{N})}^{\min}(S)$. \\
		  Let $A_S$ be the minimum cost arborescence spanning $S$ and the root in $\mathcal{N}$.
By definition of an arborescence there is exactly one directed path from the root $\rho$ to any other vertex $v \in V(A_S)$. This implies that $A_S$ contains at most one reticulation edge for each reticulation node $r \in R(\mathcal{N})$, but never both reticulation edges directed into $r \in R(\mathcal{N})$. 
If we now suppress nodes of both indegree 1 and outdegree 1 in $A_S$ and add the weights of the edges which are merged into one edge by doing so, we retrieve a directed acyclic graph $A_S'$, which contains the taxa in $S$ and whose weight equals the weight of $A_S$.
By the construction of $A_S'$, however, $A_S'$ must be a sub-arborescence of some embedded tree $\mathcal{T}_{A_{S}} \in \mathsf{T}(\mathcal{N})$, where the set of embedded trees is obtained by deleting one of the reticulation edges for each reticulation node and suppressing the resulting nodes of indegree 1 and outdegree 1, and every combination of doing so results in a phylogenetic $X$-tree (because we have assumed $\vert \TN \vert = 2^k$).  
Thus, by definition of $PD$ for trees, the weight of $A_S$ equals $PD_{\mathcal{T}_{A_S}}(S)$ and as $\mathcal{T}_{A_{S}}$ is embedded in $\mathcal{N}$ we have
$$PND(S) = PD_{\mathcal{T}_{A_S}} (S) \geq \min\limits_{\mathcal{T} \in \TN} \{ PD_{\mathcal{T}}(S) \} = PD_{\mathsf{T}(\mathcal{N})}^{\min}(S).$$ 
	Combining the above, we have $PND(S) =  PD_{\mathsf{T}(\mathcal{N})}^{\min}(S)$ as claimed.
	\end{enumerate}
\end{proof}

\noindent Comparing $PND^{inh}(S)$ and $PD_{\TN}^{\varnothing_{inh}}(S)$ of a subset $S \subseteq X$ of taxa, we see that these values, again, follow a related principle. While $PND^{inh}(S)$ considers all spanning arborescences in the network, $PD_{\TN}^{\varnothing_{inh}}(S)$ considers the spanning arborescences in each of the trees displayed by $\mathcal{N}$.
If $\vert \TN \vert = 2^k$, the two values coincide (proof similar to the proof of Proposition \ref{PND_equality}). However, in general, $PND^{inh}(S) \neq PD_{\TN}^{\varnothing_{inh}}(S)$, in particular we cannot guarantee $PND^{inh}(S) \leq PD_{\TN}^{\varnothing_{inh}}(S)$ as in Proposition \ref{PND_equality}. Consider for example the phylogenetic network $\mathcal{N}$ depicted in Figure \ref{fig_embedded} and set $S=\{A,C\}$. Then we have 
$$ PD_{\TN}^{\varnothing_{inh}}(S) = \frac{6}{11} \cdot 5 + \frac{2}{11} \cdot 6 + \frac{3}{11} \cdot 5 = \frac{57}{11},$$
but 
$$ PND^{inh}(S) = \frac{3}{4} \cdot 5 + \frac{1}{4} \cdot 6 = \frac{21}{4}.$$
Thus, $ PND^{inh}(S) > PD_{\TN}^{\varnothing_{inh}}(S)$.

\subsection{LSA associated phylogenetic diversity}
As it can be difficult to determine the set of phylogenetic $X$-trees displayed by a network $\mathcal{N}$ on $X$, we now consider the LSA tree associated with a network. The LSA tree can be seen as a way to summarize the treelike content of a phylogenetic network, on which all its embedded trees agree, without explicitly having to consider these trees.

%Definition LSA associated phylogenetic diversity
\begin{definition}[LSA associated phylogenetic diversity]
Let $\mathcal{N}$ be a rooted phylogenetic network on some taxon set $X$. 
Let $S \subseteq X$ be a subset of taxa. 
Then we define the \emph{LSA associated phylogenetic diversity} $PD^{LSA}(S)$ as
\begin{equation}
PD^{LSA}(S) \coloneqq PD_{T_{LSA}(\mathcal{N})}(S),
\end{equation}
where $PD_{\mathcal{T}_{LSA}(\mathcal{N})} (S)$ is the phylogenetic diversity of $S$ in the LSA tree $\mathcal{T}_{LSA}(\mathcal{N})$ associated with $\mathcal{N}$.
\end{definition}

%Example LSA associated phylogentic diversity
\begin{example}
Consider the rooted phylogenetic network $\mathcal{N}$ and its associated LSA tree $\mathcal{T}_{LSA}(\mathcal{N})$ depicted in Figure \ref{fig_lsatree}. Exemplarily, we set $S=\{A,B\}$ and retrieve the \emph{LSA associated phylogenetic diversity} of $S$ as $PD^{LSA}(S) = 2 + 2 + 1 = 5$.
\end{example} 

\noindent
We have introduced a variety of ways to define the phylogenetic diversity of a subset $S \subseteq X$ of taxa in a network.
However, the information about the phylogenetic diversity of a subset $S \subseteq X$ of taxa in itself is not very useful for taxon prioritization decisions.
Thus, we now turn our attention towards the generalization of phylogenetic diversity indices from trees to networks.

\section{Generalization of phylogenetic diversity indices}
After proposing different ways of generalizing the concept of phylogenetic diversity from trees to networks, we will now turn our attention to the Fair Proportion Index and the Shapley Value, two prioritization indices used in biodiversity conservation.
Even though the Fair Proportion Index and the Shapley Value are equivalent for rooted phylogenetic trees (\citet{Fuchs2015}), they differ significantly in their definition and computation.
While the Fair Proportion Index is directly based on a given rooted phylogenetic tree (cf. Definition \ref{def_fp}), the definition of the Shapley Value is based on the phylogenetic diversity of subsets of taxa, and thus, only indirectly on a given phylogenetic tree (cf. Definition \ref{def_sv}). 
To be precise, the calculation of the Shapley Value involves two steps:
	\begin{enumerate}
	\item Calculation of the phylogenetic diversity for all subsets of taxa based on a given phylogenetic tree.
	\item Calculation of the Shapley Value for all taxa based on the phylogenetic diversity calculated in step 1. 
	\end{enumerate}
This implies that we have two possibilities when extending the Shapley Value from trees to networks: We can either use any generalized definition of phylogenetic diversity (e.g. the phylogenetic net diversity, the embedded phylogenetic diversity or the LSA associated phylogenetic diversity) introduced above and calculate the Shapley Value based on this measure, or we can reduce the network to its treelike content (e.g. via the set of embedded trees or the LSA tree) and calculate the Shapley Value based on these trees.
We will, however, start with the reduction of a network to its treelike content, which is also used to generalize the Fair Proportion Index to networks.

\subsection{Embedded Shapley Value and Fair Proportion Index}
Similar to the embedded phylogenetic diversity, we will now use the set $\TN$ of phylogenetic $X$-trees displayed by a network $\mathcal{N}$ on $X$ in order to define the so-called \emph{embedded Shapley Value} and the \emph{embedded Fair Proportion Index}.

%Definition embedded diversity indices
\begin{definition}[Embedded Shapley Value, embedded Fair Proportion Index]
Let $\mathcal{N}$ be a rooted phylogenetic network on some taxon set $X$ and let $\TN$ be the (multi)set of all rooted phylogenetic $X$-trees displayed by $\mathcal{N}$. 
Then we use $DI_{\TN}^{\ast}(a)$ with $DI \in \{SV, FP\}$ to denote the embedded Shapley Value or embedded Fair Proportion Index of a taxon $a \in X$, where $\ast$ stands for $\min, \max, \sum, \varnothing$ and define 
\begin{align}
DI_{\TN}^{\min}(a) &\coloneqq \min_{\mathcal{T} \in \TN} \{ DI_{\mathcal{T}}(a) \}, \\
DI_{\TN}^{\max}(a) &\coloneqq \max_{\mathcal{T} \in \TN} \{ DI_{\mathcal{T}}(a) \}, \\
DI_{\TN}^{\sum}(a) &\coloneqq \sum_{\mathcal{T} \in \TN} DI_{\mathcal{T}}(a) \text{ and } \\
DI_{\TN}^{\varnothing}(a) &\coloneqq \frac{1}{\vert \TN \vert} \sum_{\mathcal{T} \in \TN} DI_{\mathcal{T}}(a), 
\end{align}
where $\vert \TN \vert$ is the number of phylogenetic $X$-trees displayed by $\mathcal{N}$.
If inheritance probabilities are given for $\mathcal{N}$, we also consider
\begin{align}
DI_{\TN}^{\varnothing_{inh}}(a) &\coloneqq  \sum_{\mathcal{T} \in \TN} \mathbb{P}(\mathcal{T}) \cdot PD_{\mathcal{T}}(a) \text{ and } \\
DI_{\TN}^{ML}(a) &\coloneqq PD_{\mathcal{T}'}(a) \text{ with } \mathcal{T}' = \argmax_{\mathcal{T} \in \TN} \mathbb{P}(\mathcal{T}),
\end{align}
where $\mathbb{P}(\mathcal{T})$ is the probability of $\mathcal{T}$ and $\mathcal{T}'$ is a most likely embedded tree. 
If the argmax is not unique, we arbitrarily choose one of the embedded trees with maximum probability.
\end{definition}

\noindent
Note that as the Shapley Value and the Fair Proportion Index are equivalent on rooted phylogenetic trees (\citet{Fuchs2015}), the embedded values coincide as well, i.e. $SV_{\TN}^{\min}(a) = FP_{\TN}^{\min}(a)  \text{ for all }  a \in X$ etc.

%Example embedded phylogenetic diversity
\begin{example}
Consider the rooted phylogenetic network $\mathcal{N}$ on $X=\{A,B,C,D\}$ and its embedded trees $\mathcal{T}_1, \mathcal{T}_2$ and $\mathcal{T}_3$ depicted in Figure \ref{fig_embedded} and fix taxon $A \in X$. Then we have $FP_{\mathcal{T}_1}(A) = \frac{7}{3}, FP_{\mathcal{T}_2}(A) = \frac{5}{2}$ and $FP_{\mathcal{T}_3}(A) = \frac{11}{6}$.
Moreover, $\mathbb{P}(\mathcal{T}_1) = \frac{6}{11}, \, \mathbb{P}(\mathcal{T}_2) = \frac{2}{11}$ and $\mathbb{P}(\mathcal{T}_3) = \frac{3}{11}$.
 Thus, we retrieve the different versions of the \emph{embedded Fair Proportion Index} of $A$ as
$FP_{\TN}^{\min}(A) = \frac{11}{6}, \, FP_{\TN}^{\max}(A) = \frac{5}{2}, \, FP_{\TN}^{\sum}(A) = \frac{20}{3}, \, FP_{\TN}^{\varnothing}(A) = \frac{20}{9}, \, FP_{\TN}^{\varnothing_{inh}}(A) = \frac{49}{22}$ and $FP_{\TN}^{ML}(A) = \frac{7}{3}$.
\end{example}

\subsection{LSA associated Shapley Value and Fair Proportion Index}
An alternative way of reducing a phylogenetic network to its treelike content is the LSA tree. Thus, we will now introduce the \emph{LSA associated Shapley Value} and the \emph{LSA associated Fair Proportion Index}.

%Definition LSA associated diversity indices
\begin{definition}[LSA associated Shapley Value, LSA associated Fair Proportion Index]
Let $\mathcal{N}$ be a rooted phylogenetic network on some taxon set $X$.
Let $a \in X$ be a taxon in $X$.
Then we use $DI^{LSA}(a)$ with $DI \in \{SV,FP\}$ to denote the \emph{LSA associated Shapley Value} or \emph{LSA associated Fair Proportion Index} and define
\begin{equation}
DI^{LSA}(a) \coloneqq DI_{\mathcal{T}_{LSA}(\mathcal{N})}(a),
\end{equation}
where $DI_{\mathcal{T}_{LSA}(\mathcal{N})}(a)$ is the respective diversity index (i.e. the Shapley Value or the Fair Proportion Index) in the LSA tree $\mathcal{T}_{LSA}(\mathcal{N})$ associated with $\mathcal{N}$.
\end{definition} 

\noindent
Obviously, $SV^{LSA}(a) = FP^{LSA}(a) \text{ for all } a \in X$, because the two values coincide for rooted phylogenetic trees, thus they coincide in particular for the LSA tree.

%Example LSA associated Fair Proportion Index
\begin{example}
Consider the rooted phylogenetic network $\mathcal{N}$ and its associated LSA tree $\mathcal{T}_{LSA}(\mathcal{N})$ depicted in Figure \ref{fig_lsatree} and fix taxon $a \in X$. Then the \emph{LSA associated Fair Proportion Index} of $A$ is $FP^{LSA}(A) = \frac{1}{2} + \frac{2}{1} = \frac{5}{2}$.
\end{example}

\subsection{Generalized Shapley Value}
As the definition of the Shapley Value is only indirectly based on a given phylogenetic $X$-tree and just requires a measure of phylogenetic diversity for all subsets $S \subseteq X$ of taxa (cf. Definition \ref{def_sv}), we now introduce an alternative way of calculating the Shapley Value for the taxa of a phylogenetic network $\mathcal{N}$. 
We suggest to calculate the Shapley Value according to its definition and use any measure of generalized phylogenetic diversity (e.g. the phylogenetic net diversity, the embedded phylogenetic diversity or the LSA associated phylogenetic diversity) as an input. We call the resulting value the \emph{generalized original Shapley Value}.

%Definition generalized Shapley Value
\begin{definition}[Generalized Shapley Value] \label{def_gen_sv}
Let $\mathcal{N}$ be a rooted phylogenetic network on some taxon set $X$ and let $\TN$ be the (multi)set of all rooted phylogenetic $X$-trees displayed by $\mathcal{N}$. 
Let $a \in X$ be a taxon in $X$ and let $\mathcal{PD}(S)$ denote any generalized measure of phylogenetic diversity of a subset $S \subseteq X$ of taxa in $\mathcal{N}$, i.e. $\mathcal{PD}(S) \in \{ PND(S), \, PND^{inh}(S), \, PND^{ML}(S), \, PSD(S), PD_{\TN}^{\min}(S), \, \\ PD_{\TN}^{\max}(S), \, PD_{\TN}^{\sum}(S), \, PD_{\TN}^{\varnothing}(S), \, PD_{\TN}^{\varnothing_{inh}}, \, PD_{\TN}^{ML}(S), \, PD^{LSA}(S) \}$. \\
Then we define the \emph{generalized original Shapley Value} of $a$ as
	\begin{equation}\label{gen_sv_org}
	SV_{\mathcal{PD}}(a) = \frac{1}{n!}\sum_{\substack{S \subseteq X \\ a \in S}} \Big( (\lvert S \rvert -1)!(n- \lvert S \rvert)! 
	(\mathcal{PD}(S)-\mathcal{PD}(S \setminus \{a\})) \Big),
	\end{equation}
	where $n = \vert X \vert$ and $S$ denotes a subset of species containing taxon $a$ and the sum runs over all such subsets possible. 
\end{definition}

%Example generalized Shapley Value
\begin{example}
Consider the rooted phylogenetic network $\mathcal{N}$ on $X=\{A,B,C,D\}$ depicted in Figure \ref{fig_embedded}. We now calculate the \emph{generalized original Shapley Value} of taxon $A \in X$ and choose the phylogenetic net diversity (cf. Definition \ref{def_pnd}) as input. We have to consider the following subsets $S \subseteq X$: $\{A\}, \{A,B\}, \{A,C\}, \{A,D\}, \\ \{A,B,C\}, \{A,B,D\}, \{A,C,D\}$ and $\{A,B,C,D\}$. Thus,
\begin{align*}
	SV_{PND}(A)&= \frac{1}{4!}\sum_{\substack{S \subseteq X \\ A \in S}} \Big( (\lvert S \rvert -1)!(\lvert X \rvert -\lvert S \rvert)! (PND(S)-PND(S\setminus \{A\})) \Big) \\
		   &= \frac{1}{4!} \Big[(1-1)!(4-1)!(3-0)  \\
		  	&\qquad {} + (2-1)!(4-2)! \big( (4-3)+(5-3)+(6-3) \big)  \\
		  	&\qquad {} + (3-1)!(4-3)! \big( (6-4)+(7-6)+(7-4) \big) \\
		  	&\qquad {} + (4-1)!(4-4)! (8-7) \Big] \\
		  &= \frac{1}{24} \Big[ 1 \cdot 6 \cdot 3 + 1 \cdot 2 \cdot (1+2+3) + 2 \cdot 1 \cdot (2+1+3) + 6 \cdot 1 \cdot 1 \Big] \\
		  &= \frac{48}{24} \\
		  &= 2.
	\end{align*} 
\end{example}

\subsection{Relationship between the different versions of the Shapley Value for phylogenetic networks}
We now shortly compare the generalized Shapley Value and the embedded Shapley Value of a phylogenetic network $\mathcal{N}$ on $X$. \\

\noindent
The first observation to make is that, in general, 
	\begin{itemize}
	\item $SV_{PD_{\TN}^{\min}}(a) \neq SV_{\TN}^{\min}(a)$ and
	\item $SV_{PD_{\TN}^{\max}}(a) \neq SV_{\TN}^{\max}(a)$
	\end{itemize}
for $a \in X$. 
Consider for example the rooted phylogenetic network $\mathcal{N}$ on $X=\{A,B,C,D\}$ depicted in Figure \ref{fig_embedded} and fix taxon $A$. Then we have $SV_{PD_{\TN}^{\min}}(A) = \frac{9}{4} \neq \frac{11}{6} = SV_{\TN}^{\min}(A)$ and $SV_{PD_{\TN}^{\max}}(A) = \frac{13}{6} \neq \frac{5}{2} =  SV_{\TN}^{\max}(A)$. \\

\noindent 
The second observation to make is
$$ SV_{PD_{\TN}^{ML}}(a) = SV_{\TN}^{ML}(a) $$
if the most likely tree $\mathcal{T}' \in \TN = \argmax\limits_{\mathcal{T} \in \TN} \mathbb{P}(\mathcal{T})$ is fixed, because:
	\begin{align*}
	SV_{\TN}^{ML}(a) &= SV_{\mathcal{T}'} \; \text{ with } \mathcal{T}' = \argmax\limits_{\mathcal{T} \in \TN}  \mathbb{P}(\mathcal{T}) \\
	&= \frac{1}{n!}\sum_{\substack{S \subseteq X \\ a \in S}} \Big( (\lvert S \rvert -1)!(n- \lvert S \rvert)! 
	(PD_{\mathcal{T}'}(S)-PD_{\mathcal{T}'}(S \setminus \{a\})) \Big) \\
	&= SV_{PD_{\TN}^{ML}}(a).
	\end{align*}

\noindent
Moreover, it is easy to see that for all $a  \in X$
	\begin{enumerate}
	\item [(i)] $SV_{PD_{\TN}^{\sum}}(a) = SV_{\TN}^{\sum}(a)$,
	\item [(ii)] $SV_{PD_{\TN}^{\varnothing}}(a) = SV_{\TN}^{\varnothing}(a)$ and
	\item [(iii)]$SV_{PD_{\TN}^{\varnothing_{inh}}}(a) = SV_{\TN}^{\varnothing_{inh}}(a)$.
	\end{enumerate}
	
\begin{proof}
We only show (i), but (ii) and (iii) follow analogously. \\
Recall that $PD_{\TN}^{\sum}(S) = \sum\limits_{\mathcal{T} \in \TN} PD_{\mathcal{T}}(S)$. Thus,
\begin{align*}
SV_{PD_{\TN}^{\sum}}(a) &= \frac{1}{n!}\sum_{\substack{S \subseteq X \\ a \in S}} \Big( (\lvert S \rvert -1)!(n- \lvert S \rvert)! 
	(PD_{\TN}^{\sum}(S)-PD_{\TN}^{\sum}(S \setminus \{a\})) \Big) \\
	&= \frac{1}{n!}\sum_{\substack{S \subseteq X \\ a \in S}} \Big( (\lvert S \rvert -1)!(n- \lvert S \rvert)! 
	\big(\sum\limits_{\mathcal{T} \in \TN} PD_{\mathcal{T}}(S)-\sum\limits_{\mathcal{T} \in \TN} PD_{\mathcal{T}}(S \setminus \{a\}) \big) \Big) \\
	&= \frac{1}{n!}\sum_{\substack{S \subseteq X \\ a \in S}} \Big( (\lvert S \rvert -1)!(n- \lvert S \rvert)! 
	\big(\sum\limits_{\mathcal{T} \in \TN} (PD_{\mathcal{T}}(S)-PD_{\mathcal{T}}(S \setminus \{a\})) \big) \Big). \\
\end{align*}
On the other hand we have
\begin{align*}
SV_{\TN}^{\sum}(a) &= \sum\limits_{\mathcal{T} \in \TN} SV_{\mathcal{T}} (a) \\
&= \sum\limits_{\mathcal{T} \in \TN} \Big( \frac{1}{n!}\sum_{\substack{S \subseteq X \\ a \in S}} \big( (\lvert S \rvert -1)!(n- \lvert S \rvert)! (PD_{\mathcal{T}}(S)-PD_{\mathcal{T}}(S \setminus \{a\})) \big) \Big) \\
%&= \frac{1}{n!} \sum\limits_{\mathcal{T} \in \TN} \sum_{\substack{S \subseteq X \\ a \in S}} \Big( (\lvert S \rvert -1)!(n- \lvert S \rvert)! (PD_{\mathcal{T}}(S)-PD_{\mathcal{T}}(S \setminus \{a\})) \Big) \\
%&= \frac{1}{n!} \sum_{\substack{S \subseteq X \\ a \in S}} \sum\limits_{\mathcal{T} \in \TN} \Big( (\lvert S \rvert -1)!(n- \lvert S \rvert)! (PD_{\mathcal{T}}(S)-PD_{\mathcal{T}}(S \setminus \{a\})) \Big) \\
&= \frac{1}{n!} \sum_{\substack{S \subseteq X \\ a \in S}} \Big( (\lvert S \rvert -1)! (n- \lvert S \rvert)! \big( \sum\limits_{\mathcal{T} \in \TN} (PD_{\mathcal{T}}(S)-PD_{\mathcal{T}}(S \setminus \{a\})) \big) \Big).
\end{align*}
Thus, $$SV_{PD_{\TN}^{\sum}}(a) = SV_{\TN}^{\sum}(a).$$
\end{proof}

\noindent
If we compare the LSA associated Shapley Value $SV^{LSA}$ and the generalized Shapley Value $SV_{PD^{LSA}}$ that uses the LSA associated phylogenetic diversity as input, we see that all calculations are based upon the LSA tree associated with a network $\mathcal{N}$ on $X$, thus for all $a \in X$
	\begin{enumerate}
	\item [(iii)] $SV^{LSA}(a) = SV_{PD^{LSA}}(a)$.
	\end{enumerate}

\subsection{Net Fair Proportion Index}
Before turning to our software tool and real data, we introduce one last index concept for networks, namely the \emph{Net Fair Proportion Index}. While in the previous sections we have always reduced a network $\mathcal{N}$ on $X$ to its treelike content in order to calculate the Fair Proportion Index for its taxa (i.e. we have defined the embedded Fair Proportion Index and the LSA associated Fair Proportion Index), we now try to directly adapt the definition of the Fair Proportion Index (cf. Definition \ref{def_fp}) to networks by considering \emph{all} paths between the root and a taxon. \\
Without loss of generality we assume the network $\mathcal{N}$ to come with inheritance probabilities (if no inheritance probabilities are given for $\mathcal{N}$, we set $\gamma_e = \frac{1}{2}$ for all reticulation edges $e$). \\
The idea is now to define the Net Fair Proportion Index of a taxon $a \in X$ by considering all paths from the root to $a$ and calculating a value for each path individually. Similar to the original Fair Proportion Index, we calculate this value as a weighted sum of branch lengths, where each branch length is weighted according to the number of its descendants. However, we additionally weight the possible descendants of an edge $e$ by their probability of actually being a descendant of this edge. 
We then use the weighted mean of these values for all paths, where a path is weighted according to its probability, and call the resulting value the Net Fair Proportion Index. 
%%%
\begin{definition}[Net Fair Proportion Index] \label{def_NFP}
Let $\mathcal{N}$ be a rooted phylogenetic network on some taxon set $X$. 
Let $\lambda_e$ denote the length of an edge $e$ in $\mathcal{N}$ and let $D_e$ denote the set of leaves that are descendants of $e$. \\
For each leaf $d \in D_e$ we use $\mathbb{P}_{desc}^{e}(d)$ to denote the probability of $d$ being descendent from $e$ and calculate $\mathbb{P}_{desc}^{e}(d)$ as 
	\begin{equation}
	\mathbb{P}_{desc}^{e}(d) = \sum\limits_{P \in \mathcal{P}_{e,d}} \mathbb{P}(P),
	\end{equation}
where $\mathcal{P}_{e,d}$ is the set of paths from the endpoint of $e$ to the leaf $d$ in $\mathcal{N}$ and $\mathbb{P}(P)$ is the probability of any such path (the probability of a path is calculated as the product of all probabilities assigned to its edges). \\
Now let $a \in X$ be a taxon of $\mathcal{N}$ and let $\mathcal{P}_{\rho a}$ be the set of all paths from $\rho$ to $a$ in $\mathcal{N}$. Then we define the \emph{Net Fair Proportion Index} of $a$ as
	\begin{equation}
	NFP(a) = \sum_{P \in \mathcal{P}_{\rho a}} \mathbb{P}(P) \cdot \Big( \sum_{e \in P} \frac{\lambda_e}{\sum\limits_{d \in D_e} \mathbb{P}_{desc}^{e}(d)} \Big).
	\end{equation}
\end{definition}
%%%
\begin{example} \label{example_nfp}
Consider the rooted phylogenetic network $\mathcal{N}$ on $X=\{A,B,C,D\}$ depicted in Figure \ref{fig_embedded}. 
We now calculate the \emph{Net Fair Proportion Index} for taxon $B \in X$: \\
	There are two paths from the root $\rho$ to $B$ in $\mathcal{N}$, namely 
	\begin{align*}
	P_1 &= \big( (\rho,v), (v,u), (u,r_1), (r_1,B) \big) \; \text{ with probability } \mathbb{P}(P_1) = \frac{1}{3} \; \text{ and }\\
	P_2 &= \big( (\rho,v), (v,w), (w,r_1), (r_1,B) \big) \; \text{ with probability } \mathbb{P}(P_2) = \frac{2}{3}.
	\end{align*}		
	Consider, for example, the edge $e=(\rho,v)$. The set of possible descendants from $e$ consists of the taxa $A, B$ and $C$, thus, $D_e = \{A,B,C\}$. The probabilities of these taxa descending from $e$ calculate as
		\begin{align*}
		\mathbb{P}_{desc}^{e}(A) &= 1, \\
		\mathbb{P}_{desc}^{e}(B) &= \frac{1}{3} + \frac{2}{3} = 1 \text{ and } \\
		\mathbb{P}_{desc}^{e}(C) &= \frac{3}{4}.
		\end{align*}
	Analogously, these probabilities can be calculated for all other edges on $P_1$ and $P_2$. 
	Omitting edges of length 0 (i.e. hybridization edges) in the sum, we have
		\begin{align*}
		NFP(B) &= \frac{1}{3} \Big( \underbrace{\frac{1}{1}}_{(r_1,B)} + \underbrace{\frac{1}{\underbrace{1}_{A} +  \underbrace{\frac{1}{3}}_{B}}}_{(v,u)} +  \underbrace{\frac{1}{\underbrace{1}_{A} + \underbrace{1}_{B} + \underbrace{\frac{3}{4}}_{C}}}_{(\rho,v)} \Big) \\
		&\hspace{5mm} + \frac{2}{3} \Big( \underbrace{\frac{1}{1}}_{(r_1,B)} + \underbrace{\frac{1}{\underbrace{\frac{2}{3}}_{B} +  \underbrace{\frac{3}{4}}_{C}}}_{(v,w)} +  \underbrace{\frac{1}{\underbrace{1}_{A} + \underbrace{1}_{B} + \underbrace{\frac{3}{4}}_{C}}}_{(\rho,v)} \Big) \\
		&= \frac{1}{3} \cdot \frac{93}{44} + \frac{2}{3} \cdot \frac{387}{187} \\
		&= \frac{1559}{748} \\
		&\approx 2.08.
		\end{align*}
Similar calculations yield
	\begin{align*}
	NFP(A) &= \frac{93}{44} \approx 2.11, \\
	NFP(C) &= \frac{2059}{935} \approx 2.20 \text{ and } \\
	NFP(D) &= \frac{13}{5} = 2.6.
	\end{align*}
Note that
	\begin{align*}
	\sum_{a \in X} NFP(a) &= \frac{93}{44} + \frac{1559}{748} + \frac{2059}{935} + \frac{13}{5} \\
	&= 9,
	\end{align*}
thus, the sum of the Net Fair Proportion Indices equals the sum of edge lengths in $\mathcal{N}$. 
\end{example}
%%%
\begin{remarks} \leavevmode
	\begin{itemize}
	\item By definition of the Net Fair Proportion Index, this measure is \emph{efficient}, i.e. 
	$$ \sum_{a \in X} NFP(a) = weight(\mathcal{N}), $$
	where $weight(\mathcal{N})$ is the sum of branch lengths of the rooted phylogenetic network $\mathcal{N}$ on $X$.
	\item For a phylogenetic $X$-tree $\mathcal{T}$, the Net Fair Proportion Index reduces to the original Fair Proportion Index, i.e. for all $a \in X$ 
	$$ NFP(a) = FP(a).$$
	\end{itemize}
\end{remarks}

\section{Software and Data}
In order to calculate the different generalized measures of phylogenetic diversity and generalized diversity indices introduced above, we developed a software tool called NetDiversity, which is available from \\ www.mareikefischer.de/Software/NetDiversity.zip.
The tool is written in the programming language Perl and uses modules from BioPerl (\citet{BioPerl}), in particular the Bio::PhyloNetwork package (\citet{Cardona2008a})
The program takes networks represented in the so-called extended Newick format (\citet{Cardona2008}) as an input. Depending on the options chosen, the program either outputs any measure of generalized phylogenetic diversity for all subsets of taxa or any generalized diversity index for all taxa of the network. However, currently the tool can only calculate measures independent of inheritance probabilities.

\noindent
We now apply NetDiversity to a phylogenetic network of swordtails and platyfishes (\textit{Xiphophorus}: Poeciliidae) (cf. \citet{Solis-Lemus2016}). This is one of the few published hybridization networks, even though hybridization is suspected to have occurred in a variety of other organisms as well.
The \textit{Xiphophorus} hybridization network inferred in \citet{Solis-Lemus2016} contains $24$ species and $2$ reticulation nodes (cf. Figure \ref{figure_fishnet}). 
Exemplarily, we use NetDiversity to calculate the different versions of the Fair Proportion Index for the \textit{Xiphophorus} species. Note that there are $2^{24} = 16777216$ possible subsets of taxa for a network on $24$ species, which is why we refrain from calculating any measure of generalized phylogenetic diversity for all subsets of \textit{Xiphophorus} or the generalized Shapley value here.
Table \ref{fp_fish} summarizes the results. For the \textit{Xiphophorus} network, the rankings obtained by the embedded Fair Proportion Indices and the LSA associated Fair Proportion Index are very similar. There are, however, two striking differences concerning the species \textit{X. xiphidium} and \textit{X. nezahuacoyotl}. 
While \textit{X. xiphidium} is ranked low by $FP_{\TN}^{\min}$, it is placed among the top 10 species by all other indices.
The other difference between the indices concerns \textit{X. nezahuacoyotl}, a hybrid species. \textit{X. nezahuacoyotl} is ranked first by $FP^{LSA}$, while it is ranked $12^{th}$, $12^{th}$ and $15^{th}$ by the other indices.

Thus, in case of the \textit{Xiphophorus} network, the different versions of the generalized Fair Proportion Index yield similar results, but there are striking differences. In particular the question of whether hybrid species are of high or low importance for overall biodiversity remains to be considered from a biological perspective. 

%Figure of the fish network
\begin{figure}[htbp]
	\centering
	\includegraphics[scale=0.2]{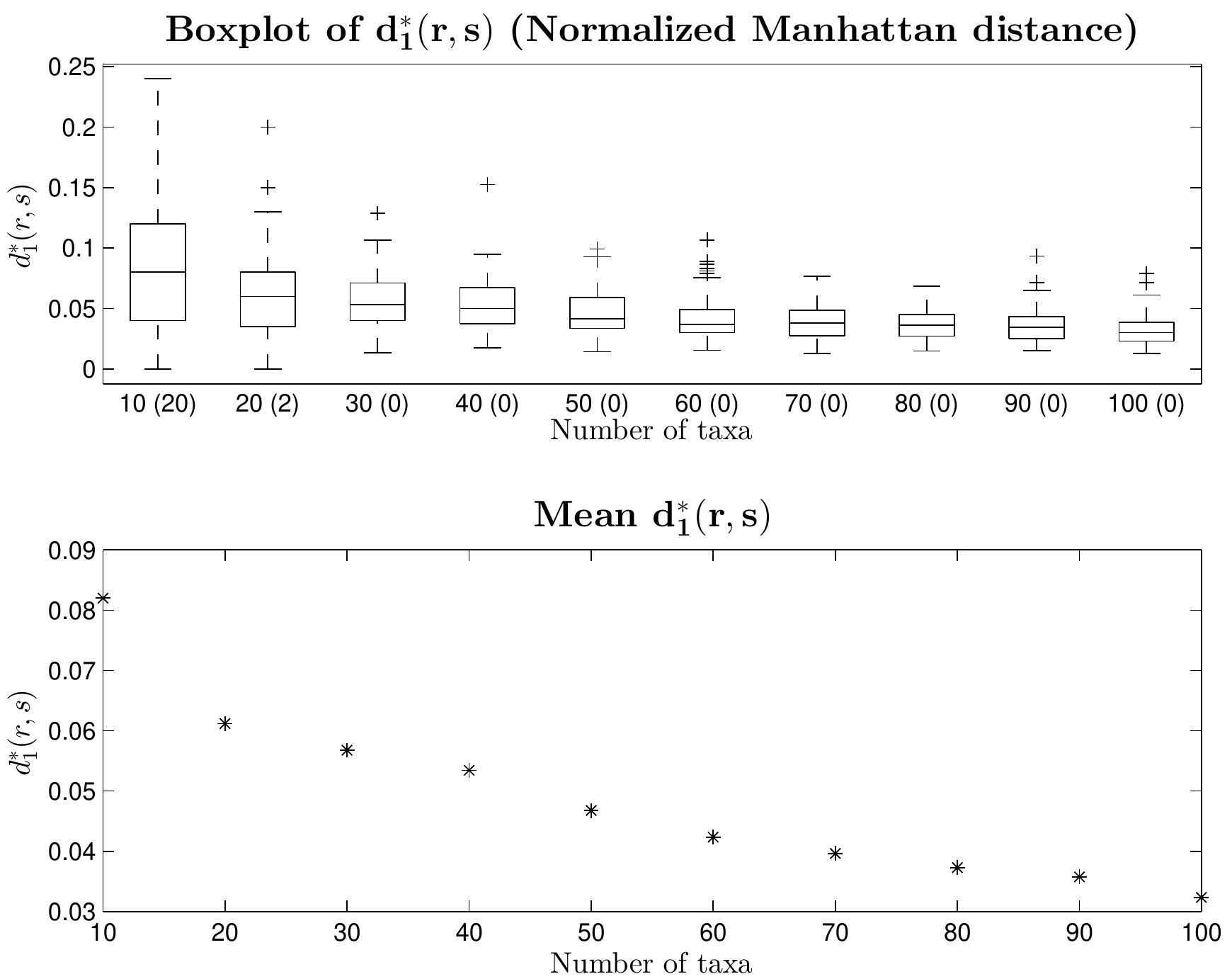}
	\caption{\textit{Xiphophorus} hybridization network with $24$ species and $2$ reticulation nodes (see supporting information (S1 Text) for more information; Figure created with Dendroscope (\citet{Huson2012}).}
	\label{figure_fishnet}
\end{figure}

%Fair Proportion Indices
\begin{table}
\caption{Embedded and LSA associated Fair Proportion Indices (rounded) for the \textit{Xiphophorus} species. The numbers in brackets indicate how species are ranked by the different indices.}
\label{fp_fish}
\begin{tabular}{lllll}
\hline\noalign{\smallskip}
 & $FP_{\TN}^{\min}$ & $FP_{\TN}^{\max}$ & $FP_{\TN}^{\varnothing}$ & $FP^{LSA}$ \\
\noalign{\smallskip}\hline\noalign{\smallskip}
\textit{X. gordoni} & $1.711$ (1) & $1.879$ (3) & $1.795$ (3) & $1.879$ (4)\\
\textit{X. meyeri} & $1.711$ (1) & $1.879$ (3) & $1.795$ (3) & $1.879$ (4)\\
\textit{X. continens} & $1.710$ (3) & $2.117$ (1) & $1.913$ (1) & $2.047$ (2) \\
\textit{X. pygmaeus} & $1.710$ (3) & $2.117$ (1) & $1.913$ (1) & $2.047$ (2) \\
\textit{X. couchianus} & $1.580$ (5) & $1.747$ (7) & $1.663$ (5) & $1.747$ (8) \\
\textit{X. multilineatus} & $1.418$ (6) & $1.835$ (5) & $1.627$ (6) & $1.765$ (6) \\
\textit{X. nigrensis}  & $1.418$ (6) & $1.835$ (5) & $1.627$ (6) & $1.765$ (6) \\
\textit{X. birchmanni} & $1.027$ (8) & $1.341$ (9) & $1.184$ (8) & $1.271$ (10) \\
\textit{X. malinche} & $1.027$ (8) & $1.341$ (9) & $1.184$ (8) & $1.271$ (10) \\
\textit{X. monticolus} & $0.796$ (10) & $0.796$ (14) & $0.796$ (13) & $0.796$ (13) \\
\textit{X. clemenciae} & $0.796$ (10) & $0.796$ (14) & $0.796$ (13) & $0.796$ (13) \\
\textit{X. alvarezi} & $0.782$ (12) & $0.782$ (16) & $0.782$ (15) & $0.782$ (15) \\
\textit{X. mayae}  & $0.782$ (12) & $0.782$ (16) & $0.782$ (15) & $0.782$ (15) \\
\textit{X. hellerii} & $0.618$ (14) & $0.618$ (18) & $0.618$ (18) & $0.618$ (18) \\
\textbf{\textit{X. nezahuacoyotl}} & $\mathbf{0.560}$ \textbf{(15)} & $\mathbf{1.049}$ \textbf{(12)} & $\mathbf{0.804}$ \textbf{(12)} & $\mathbf{2.237}$ \textbf{(1)} \\
\textit{X. montezumae} & $0.560$ (15) & $1.060$ (11) & $0.810$ (11) & $0.990$ (12) \\
\textit{X. signum} & $0.532$ (17) & $0.532$ (20) & $0.532$ (19) & $0.532$ (20) \\
\textit{X. cortezi} & $0.525$ (18) & $0.840$ (13) & $0.682$ (17) & $0.770$ (17) \\
\textit{X. variatus} & $0.450$ (19) & $0.576$ (19) & $0.494$ (20) & $0.578$ (19) \\
\textbf{\textit{X. xiphidium}} & $\mathbf{0.305}$ \textbf{(20)} & $\mathbf{1.717}$ \textbf{(8)} & $\mathbf{1.011}$ \textbf{(10)} & $\mathbf{1.717}$ \textbf{(9)} \\
\textit{X. evelynae} & $0.248$ (21) & $0.416$ (21) & $0.332$ (21) & $0.416$ (21) \\
\textit{X. milleri} & $0.147$ (22) & $0.285$ (22) & $0.216$ (22) & $0.285$ (22) \\
\textit{X. andersi} & $0.117$ (23) & $0.218$ (23) & $0.168$ (23) & $0.218$ (23) \\
\textit{X. maculatus} & $0.079$ (24) & $0.136$ (24) & $0.108$ (24) & $0.136$ (24) \\
\noalign{\smallskip}\hline
\end{tabular}
\end{table}

\section{Discussion and Outlook}
In this paper, we have introduced different approaches towards the generalization of phylogenetic diversity and phylogenetic diversity indices from trees to networks. 
Our approaches provide an extension to existing prioritization tools in conservation biology and allow for the consideration of phylogenetic networks in prioritization decisions. 
This is of importance if the evolutionary history of a set of species is known to be non-treelike, and thus cannot be represented by a phylogenetic tree.
Here, we have mainly focused on hybridization networks, but mathematically our approaches are also applicable to networks representing horizontal gene transfer.
We have applied our methods to a phylogenetic network representing the evolutionary relationships among swordtails and platyfishes (\textit{Xiphophorus}: Poeciliidae), whose evolution is characterized by widespread hybridization. We have seen that different biodiversity indices may induce striking differences in the ranking order of taxa for conservation.
Therefore, we remark that further research concerning the biological plausibility of our approaches is necessary before they can be put into practice. 
This may be achieved when more phylogenetic networks for different groups of organisms become available and can be analyzed under both a biological and mathematical perspective.
Decisions in biodiversity conservation and taxon prioritization do always require thorough examination and should include as much information as possible.

\section*{Supporting Information}
\textbf{S1 Text. Supporting information file that contains the \textit{Xiphophorus} hybridization network (\citet{Solis-Lemus2016}, its LSA tree and its embedded trees.}

\section*{Acknowledgements}
We thank Volkmar Liebscher for helpful discussions on this research project and two anonymous reviewers for helpful comments on an earlier version of this manuscript.
The first author also thanks the Ernst-Moritz-Arndt-University Greifswald for the Landesgraduiertenförderung studentship, under which this
work was conducted.

\section*{References}
\bibliographystyle{model1-num-names}\biboptions{authoryear}
\bibliography{Sources}   % name your BibTeX data base

\begin{thebibliography}{16}
\expandafter\ifx\csname natexlab\endcsname\relax\def\natexlab#1{#1}\fi
\providecommand{\url}[1]{\texttt{#1}}
\providecommand{\href}[2]{#2}
\providecommand{\path}[1]{#1}
\providecommand{\DOIprefix}{doi:}
\providecommand{\ArXivprefix}{arXiv:}
\providecommand{\URLprefix}{URL: }
\providecommand{\Pubmedprefix}{pmid:}
\providecommand{\doi}[1]{\href{http://dx.doi.org/#1}{\path{#1}}}
\providecommand{\Pubmed}[1]{\href{pmid:#1}{\path{#1}}}
\providecommand{\bibinfo}[2]{#2}
\ifx\xfnm\relax \def\xfnm[#1]{\unskip,\space#1}\fi
%Type = Article
\bibitem[{Faith(1992)}]{Faith1992}
\bibinfo{author}{D.~P. Faith},
\newblock \bibinfo{title}{Conservation evaluation and phylogenetic diversity},
\newblock \bibinfo{journal}{Biological Conservation} \bibinfo{volume}{61}
  (\bibinfo{year}{1992}) \bibinfo{pages}{1–10}.
%Type = Article
\bibitem[{Haake et~al.(2007)Haake, Kashiwada, and Su}]{Haake2007}
\bibinfo{author}{C.-J. Haake}, \bibinfo{author}{A.~Kashiwada},
  \bibinfo{author}{F.~E. Su},
\newblock \bibinfo{title}{The {S}hapley value of phylogenetic trees},
\newblock \bibinfo{journal}{J. Math. Biol.} \bibinfo{volume}{56}
  (\bibinfo{year}{2007}) \bibinfo{pages}{479–497}.
%Type = Article
\bibitem[{Hartmann(2013)}]{Hartmann2013}
\bibinfo{author}{K.~Hartmann},
\newblock \bibinfo{title}{The equivalence of two phylogenetic biodiversity
  measures: the {S}hapley value and {F}air {P}roportion index.},
\newblock \bibinfo{journal}{J Math Biol} \bibinfo{volume}{67}
  (\bibinfo{year}{2013}) \bibinfo{pages}{1163--1170}.
%Type = Article
\bibitem[{Fuchs and Jin(2015)}]{Fuchs2015}
\bibinfo{author}{M.~Fuchs}, \bibinfo{author}{E.~Y. Jin},
\newblock \bibinfo{title}{Equality of {S}hapley value and fair proportion index
  in phylogenetic trees.},
\newblock \bibinfo{journal}{J Math Biol} \bibinfo{volume}{71}
  (\bibinfo{year}{2015}) \bibinfo{pages}{1133--1147}.
%Type = Article
\bibitem[{Wicke and Fischer(2017)}]{Wicke2015}
\bibinfo{author}{K.~Wicke}, \bibinfo{author}{M.~Fischer},
\newblock \bibinfo{title}{Comparing the rankings obtained from two biodiversity
  indices: the {F}air {P}roportion {I}ndex and the {S}hapley {V}alue},
\newblock \bibinfo{journal}{Journal of Theoretical Biology}
  \bibinfo{volume}{430} (\bibinfo{year}{2017}) \bibinfo{pages}{207--214}.
%Type = Inbook
\bibitem[{Chernomor et~al.(2016)Chernomor, Klaere, von Haeseler, and
  Minh}]{Chernomor2016}
\bibinfo{author}{O.~Chernomor}, \bibinfo{author}{S.~Klaere},
  \bibinfo{author}{A.~von Haeseler}, \bibinfo{author}{B.~Q. Minh},
  \bibinfo{title}{Split Diversity: Measuring and Optimizing Biodiversity Using
  Phylogenetic Split Networks}, \bibinfo{publisher}{Springer International
  Publishing}, \bibinfo{address}{Cham}, \bibinfo{year}{2016}, pp.
  \bibinfo{pages}{173--195}. \URLprefix
  \url{http://dx.doi.org/10.1007/978-3-319-22461-9\_9}.
  \DOIprefix\doi{10.1007/978-3-319-22461-9\_9}.
%Type = Article
\bibitem[{Volkmann et~al.(2014)Volkmann, Martyn, Moulton, Spillner, and
  Mooers}]{Volkmann2014}
\bibinfo{author}{L.~Volkmann}, \bibinfo{author}{I.~Martyn},
  \bibinfo{author}{V.~Moulton}, \bibinfo{author}{A.~Spillner},
  \bibinfo{author}{A.~O. Mooers},
\newblock \bibinfo{title}{Prioritizing {P}opulations for {C}onservation {U}sing
  {P}hylogenetic {N}etworks},
\newblock \bibinfo{journal}{PLoS ONE} \bibinfo{volume}{9}
  (\bibinfo{year}{2014}) \bibinfo{pages}{e88945}.
%Type = Article
\bibitem[{Solís-Lemus and Ané(2016)}]{Solis-Lemus2016}
\bibinfo{author}{C.~Solís-Lemus}, \bibinfo{author}{C.~Ané},
\newblock \bibinfo{title}{Inferring phylogenetic networks with maximum
  pseudolikelihood under incomplete lineage sorting},
\newblock \bibinfo{journal}{PLOS Genetics} \bibinfo{volume}{12}
  (\bibinfo{year}{2016}) \bibinfo{pages}{e1005896}.
%Type = Book
\bibitem[{Huson et~al.(2011)Huson, Rupp, and Scornavacca}]{Huson2011}
\bibinfo{author}{D.~H. Huson}, \bibinfo{author}{R.~Rupp},
  \bibinfo{author}{C.~Scornavacca}, \bibinfo{title}{Phylogenetic Networks:
  Concepts, Algorithms and Applications}, \bibinfo{publisher}{Cambridge
  University Press}, \bibinfo{address}{New York, NY, USA},
  \bibinfo{year}{2011}.
%Type = Article
\bibitem[{Lengauer and Tarjan(1979)}]{Lengauer1979}
\bibinfo{author}{T.~Lengauer}, \bibinfo{author}{R.~E. Tarjan},
\newblock \bibinfo{title}{A fast algorithm for finding dominators in a
  flowgraph},
\newblock \bibinfo{journal}{TOPLAS} \bibinfo{volume}{1} (\bibinfo{year}{1979})
  \bibinfo{pages}{121–141}.
%Type = Inbook
\bibitem[{Karp(1972)}]{Karp1972}
\bibinfo{author}{R.~M. Karp}, \bibinfo{title}{Reducibility among
  {C}ombinatorial Problems}, \bibinfo{publisher}{Springer US},
  \bibinfo{address}{Boston, MA}, \bibinfo{year}{1972}, pp.
  \bibinfo{pages}{85--103}. \URLprefix
  \url{http://dx.doi.org/10.1007/978-1-4684-2001-2\_9}.
  \DOIprefix\doi{10.1007/978-1-4684-2001-2\_9}.
%Type = Article
\bibitem[{van Iersel et~al.(2010)van Iersel, Semple, and Steel}]{Iersel2010}
\bibinfo{author}{L.~van Iersel}, \bibinfo{author}{C.~Semple},
  \bibinfo{author}{M.~Steel},
\newblock \bibinfo{title}{Locating a tree in a phylogenetic network},
\newblock \bibinfo{journal}{Information Processing Letters}
  \bibinfo{volume}{110} (\bibinfo{year}{2010}) \bibinfo{pages}{1037 -- 1043}.
%Type = Article
\bibitem[{Stajich(2002)}]{BioPerl}
\bibinfo{author}{J.~E. Stajich},
\newblock \bibinfo{title}{The {B}ioperl {T}oolkit: {P}erl {M}odules for the
  {L}ife {S}ciences},
\newblock \bibinfo{journal}{Genome Research} \bibinfo{volume}{12}
  (\bibinfo{year}{2002}) \bibinfo{pages}{1611–1618}.
%Type = Article
\bibitem[{Cardona et~al.(2008{\natexlab{a}})Cardona, Rosselló, and
  Valiente}]{Cardona2008a}
\bibinfo{author}{G.~Cardona}, \bibinfo{author}{F.~Rosselló},
  \bibinfo{author}{G.~Valiente},
\newblock \bibinfo{title}{A perl package and an alignment tool for phylogenetic
  networks},
\newblock \bibinfo{journal}{BMC Bioinformatics} \bibinfo{volume}{9}
  (\bibinfo{year}{2008}{\natexlab{a}}) \bibinfo{pages}{175}.
%Type = Article
\bibitem[{Cardona et~al.(2008{\natexlab{b}})Cardona, Rosselló, and
  Valiente}]{Cardona2008}
\bibinfo{author}{G.~Cardona}, \bibinfo{author}{F.~Rosselló},
  \bibinfo{author}{G.~Valiente},
\newblock \bibinfo{title}{Extended {N}ewick: it is time for a standard
  representation of phylogenetic networks},
\newblock \bibinfo{journal}{BMC Bioinformatics} \bibinfo{volume}{9}
  (\bibinfo{year}{2008}{\natexlab{b}}) \bibinfo{pages}{532}.
%Type = Article
\bibitem[{Huson and Scornavacca(2012)}]{Huson2012}
\bibinfo{author}{D.~H. Huson}, \bibinfo{author}{C.~Scornavacca},
\newblock \bibinfo{title}{Dendroscope 3: An {I}nteractive {T}ool for {R}ooted
  {P}hylogenetic {T}rees and {N}etworks},
\newblock \bibinfo{journal}{Systematic Biology} \bibinfo{volume}{61}
  (\bibinfo{year}{2012}) \bibinfo{pages}{1061–1067}.

\end{thebibliography}

\end{document}